\documentclass[10pt,conference]{IEEEtran}
\PassOptionsToPackage{table}{xcolor}

\usepackage{cite}
\usepackage{amsmath,amssymb,amsfonts, amsthm}
\usepackage{graphicx}
\usepackage{textcomp}
\usepackage{enumitem}
\usepackage[hyphens]{url}
\usepackage{fancyhdr}
\usepackage{hyperref}
\usepackage{multirow}
\usepackage{tikz}
\usepackage{pifont}
\usepackage[linesnumbered,ruled,vlined]{algorithm2e}
\usepackage{algpseudocode}
\usepackage{multicol}
\usepackage{booktabs}

\newtheorem{theorem}{Theorem}
\newtheorem{lemma}[theorem]{Lemma}
\usepackage{graphicx}
\usepackage{subcaption}

\title{HF-NTT: Hazard-Free Dataflow Accelerator for Number Theoretic Transform}

\author{
    \IEEEauthorblockN{Xiangchen Meng}
    \IEEEauthorblockA{
        Microelectronics Thrust\\
        The Hong Kong University of \\Science and Technology (Guangzhou) \\
        Guangzhou, China\\
        xmeng027@connect.hkust-gz.edu.cn
    }
    \and
    \IEEEauthorblockN{Zijun Jiang}
    \IEEEauthorblockA{
        Microelectronics Thrust\\
        The Hong Kong University of \\Science and Technology (Guangzhou)\\
        Guangzhou, China\\
        zjiang438@connect.hkust-gz.edu.cn
    }
    \and
    \IEEEauthorblockN{Yangdi Lyu}
    \IEEEauthorblockA{
        Microelectronics Thrust\\
        The Hong Kong University of \\Science and Technology (Guangzhou)\\
        Guangzhou, China\\
        yangdilyu@hkust-gz.edu.cn
    }
}

\begin{document}
\maketitle

\begin{abstract}
    Polynomial multiplication is one of the fundamental operations in many applications, such as fully homomorphic encryption (FHE). However, the computational inefficiency stemming from polynomials with many large-bit coefficients poses a significant challenge for the practical implementation of FHE. The Number Theoretic Transform (NTT) has proven an effective tool in enhancing polynomial multiplication, but a fast and adaptable method for generating NTT accelerators is lacking. In this paper, we introduce HF-NTT, a novel NTT accelerator. HF-NTT efficiently handles polynomials of varying degrees and moduli, allowing for a balance between performance and hardware resources by adjusting the number of Processing Elements (PEs). Meanwhile, we introduce a data movement strategy that eliminates the need for bit-reversal operations, resolves different hazards, and reduces the clock cycles. Furthermore, Our accelerator includes a hardware-friendly modular multiplication design and a configurable PE capable of adapting its data path, resulting in a universal architecture. We synthesized and implemented prototype using Vivado 2022.2, and evaluated it on the Xilinx Virtex-7 FPGA platform. The results demonstrate significant improvements in Area-Time-Product (ATP) and processing speed for different polynomial degrees. In scenarios involving multi-modulus polynomial multiplication, our prototype consistently outperforms other designs in both ATP and latency metrics.
\end{abstract}

\section{Introduction}
In the landscape of the internet, the exchange of information between clients and cloud services poses a growing threat to privacy. As user awareness of privacy protection increases, a range of protective measures have emerged to address these concerns. Fully homomorphic encryption, based on the Ring Learning With Errors (RLWE) problem, is a promising approach to ensuring the privacy and security of data exchanges between local and cloud environments. This advanced encryption framework allows for the computation of ciphertext in the cloud without the need for a decryption key.

In the fully homomorphic encryption framework, plaintext data is converted into two polynomials, each containing a large number of ($2^{10}$ to $2^{18}$) elements with coefficients of significant size (hundreds of bits). A wide range of operations, such as homomorphic addition, multiplication, and permutation, can be carried out in the cloud without needing a decryption key. Homomorphic multiplication on polynomial rings is a crucial operation and a primary bottleneck, requiring numerous sub-operations and substantial computing resources. As a result, the acceleration strategy for fully homomorphic encryption primarily focuses on optimizing multiplication. NTT is commonly used to expedite polynomial multiplication by converting polynomials from coefficient notation to point-valued notation~\cite{2016Speeding}. 
\begin{figure}[t]
    \centering
    \includegraphics[width=0.98\columnwidth]{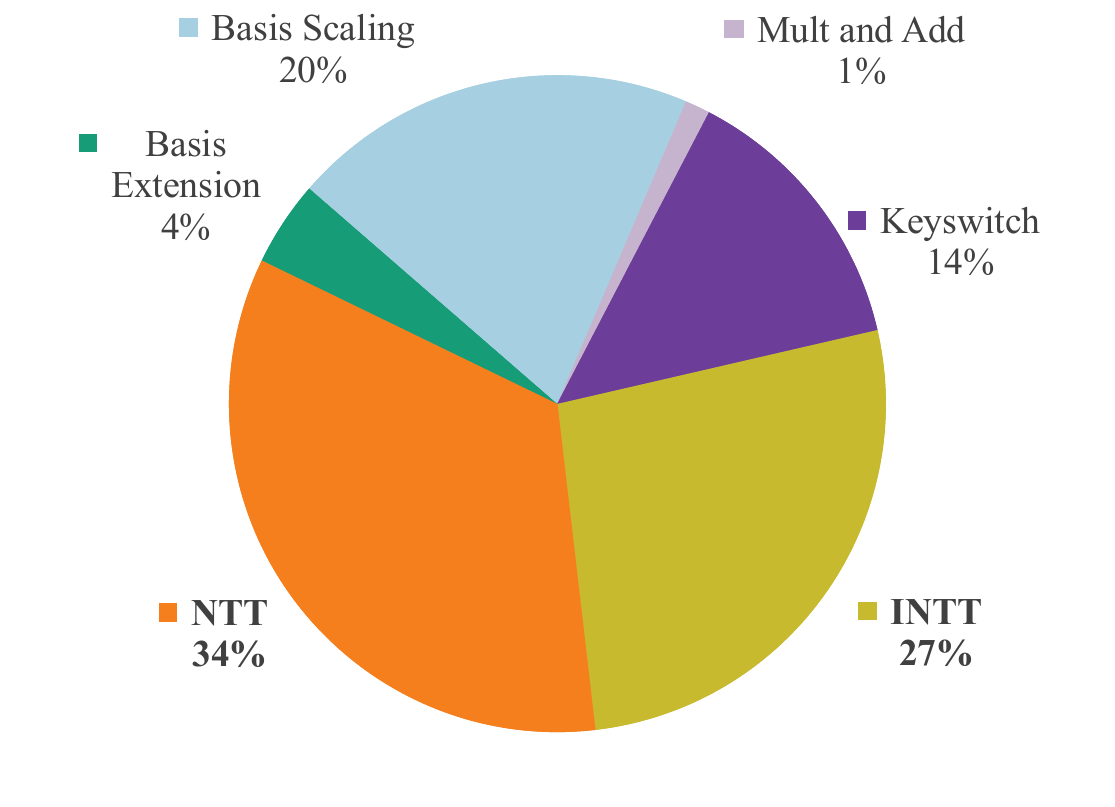}
    \caption{The Time Breakdown of Homomorphic Multiplication (Polynomial Degree $N = 4096, \left\lceil\log q \right\rceil = 32$)}
    \label{fig: The Time Breakdown of Homomorphic Multiplication}
\end{figure}
Figure~\ref{fig: The Time Breakdown of Homomorphic Multiplication} illustrates the distribution of time for homomorphic multiplication operations on a 12th Gen Intel(R) Core(TM) i7-12700 CPU platform. The analysis reveals that NTT and its inverse (INTT) processes collectively account for over 50\% of the total execution time in homomorphic multiplication. It is noteworthy that NTT and INTT are not restricted to basic multiplication transformations. For example, the Keyswitch operation, a crucial component of the multiplication process, also involves NTT operations. As a result, the majority of FHE hardware accelerators primarily focus on optimizing the NTT operation.

Hardware acceleration platforms for NTT include FPGAs~\cite{9976314, 9943988}, GPUs~\cite{10071017, 10031464}, and ASICs~\cite{2021F1}. NTT-based hardware accelerators can process large datasets and perform intensive tasks much faster than CPUs, resulting in over 100 times speedup compared to CPU implementations~\cite{2021F1}. Recent NTT acceleration research emphasizes the reconfigurability of hardware accelerators, enabling them to adapt to various tasks. As a result, FPGAs are preferred due to their flexibility and optimal balance between performance and power efficiency. However, many hardware designs for polynomial multiplication~\cite{7002209,8866755,9114594,8928952} were limited by fixed polynomial degrees and a limited number of PEs that could not be changed, restricting their applicability to certain scenarios. To enhance versatility, Hu \textit{et al.}~\cite{9943988} proposed an area-efficient and configurable NTT-based polynomial multiplier capable of adjusting the number of PEs based on different parameters. However, it was only suitable for polynomials with small bit-width coefficients. As designs became more complex, the scalability of hardware solutions faced challenges in performance and area efficiency when dealing with larger moduli and ciphertext lengths, leading to difficulties in control logic implementation~\cite{9171507}. Additionally, addressing memory overhead posed another challenge. For example, Du \textit{et al.}~\cite{10.1145/2902961.2902969} used a Ping-Pong mechanism to build memory, resulting in complex control logic and double the storage area. Regarding control logic and memory issues, Mu \textit{et al.}~\cite{9882331} proposed an NTT accelerator design method without memory conflict. This method allows for the parameters to be defined based on the specific task. However, it failed to address the limitations of coefficients exceeding large bits in FHE, limiting its applicability.

To address the diverse requirements and reconfigurability demands of NTT accelerators, we propose HF-NTT, an innovative NTT accelerator that optimizes resource utilization and offers flexibility for a range of application scenarios.

\begin{enumerate} 
\item We introduce a new NTT data storage and movement method to eliminate memory conflicts between PEs and memory as well as stalls in PEs for accelerators with varying numbers of PEs.

\item We utilize Residue Number System (RNS) to decompose coefficients with large bit widths, enhancing hardware platform adaptability and resource utilization efficiency.

\item We introduce a configurable butterfly unit (CBU) capable of executing multiple operations, including NTT, multiplication, and Inverse NTT. We optimize the multiplier through Barrett algorithm and Digital Signal Processing (DSP) to achieve improved performance-area efficiency on FPGA.

\item Our HF-NTT accelerator, implemented on the Xilinx VIRTEX-7 FPGA platform, demonstrates high performance and area efficiency. Through experimentation and exploration of design parameters, our approach exhibits adaptability to large designs.
\end{enumerate}

The remainder of this paper is structured as follows: Section~\ref{sec:pre} introduces NTT operation and Barrett modular multiplication. Section~\ref{sec:method} provides our data storage method and data movement with hazards. Section~\ref{sec:arch} proposes a detailed description of our hardware design, HF-NTT. Section~\ref{sec:exp} analyzes experimental results, comparing our approach to state-of-the-art designs. Section~\ref{sec:conclude} concludes with contributions and results.

\section{Preliminary}
\label{sec:pre}
\subsection{Polynomial Multiplication Based on NTT}
In the context of RLWE, polynomials are defined over the ring $R_q[x] = \mathbb{Z}_q[x]/\left(x^N+1\right)$, where $\mathbb{Z}_q[x]$ represents a quotient ring with prime $q$, $N$ denotes the polynomial length, and $q$ fulfills the requirement: $q\equiv1\mod2N$. For the school-book algorithm, the polynomial multiplication between $\mathbf{a}(x) = \sum_{n=0}^{N-1}a_nx^n$ and  $\mathbf{b}(x) = \sum_{m=0}^{N-1}b_mx^m$ is described as $ \mathbf{c}(x) = \sum_{n=0}^{N-1}\sum_{m=0}^{N-1}a_nb_mx^{m+n}$ with $O\left(N^2\right)$  complexity. 

NTT is used to transform a polynomial from coefficient to point-value representation. Following point-wise multiplication, the Inverse NTT (INTT) operation is necessary to revert the value representation to its original coefficient representation. Through them reducing its complexity from $O(N^2)$ to $O(N\log N)$ to expedite polynomial multiplication. See Algorithm~\ref{alg:algorithm 1} for a detailed explanation.

\begin{algorithm}[h]
\caption{Polynomial Multiplication Based on NTT} 
\label{alg:algorithm 1}
\KwIn{Polynomials $\mathbf{a}(x),\mathbf{b}(x) \in \mathbb{Z}_q[x]/(x^N+1)$, \\
      $2N$-th primitive root of unity $\omega_{2N} \in \mathbb{Z}_q$, $N = 2^k$} 
\KwOut{$\mathbf{c}(x)=\mathbf{a}(x)\cdot\mathbf{b}(x) \in \mathbb{Z}_q[x]/(x^N+1)$}

$\mathbf{A}(x) = \sum_{n=0}^{N-1}a_n\omega_{2N}^{nk}\mod q$ \tcp*[r]{$O(N\log N)$}

$\mathbf{B}(x) = \sum_{n=0}^{N-1}b_n\omega_{2N}^{nk}\mod q$ \tcp*[r]{$O(N\log N)$}

$\mathbf{C}(x) = \mathbf{A}(x) \odot \mathbf{B}(x)$ \tcp*[r]{$O(N)$}

$\mathbf{c}(x) = N^{-1}\sum_{n=0}^{N-1}C_n\omega_{2N}^{-nk}\mod q$ \tcp*[r]{$O(N\log N)$}

\Return{$\mathbf{c}(x)$}
\end{algorithm}


In Algorithm~\ref{alg:algorithm 1}, the symbol $\odot$ denotes dot-wise multiplication, while $\omega_N^n = g^{\frac{(q-1)n}{N}}$ mod $q$ and $g$ is the primitive root of the prime $q$. The \textbf{NTT} and \textbf{INTT} adopt Cooley-Tukey (CT) and Gentlemen-Sande (GS) algorithms, to finish the iterative calculation process, which could be formulated in line 1, 2 and line 4.

\begin{figure*}[ht]
    \centering
    \includegraphics[width = 1\textwidth]{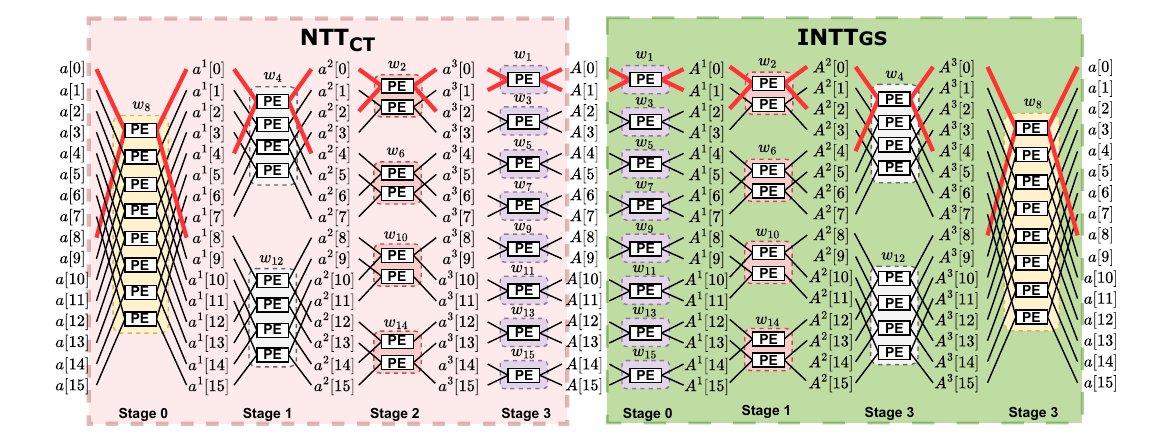}
    \caption{Dataflow of NTT and INTT Using GS and CT Algorithms ($N = 16$, $a$: coefficient representation, $A$: point-value representation)}
    \label{fig: Dataflow of NTT Operation for 16 Polynomial Degree}
\end{figure*}

\subsection{Barrett Modular Multiplication}
\begin{algorithm}[hb]
\caption{Barrett Modular Multiplication}
\label{alg:SoftwareBarrettModMul}
\KwIn{$a, b, q \in \mathbb{Z}^+$, $a, b < q$}
\KwOut{$z = (a \cdot b) \pmod{q}$}
$k \gets \lceil\log_2q\rceil$\;
$m \gets \lfloor 2^{2k} / q \rfloor$\;
$t_1 \gets a \cdot b$\;
$t_2 \gets t_1 \cdot m >> 2k$\;
$t_3 \gets t_2 \cdot q$\;
$t_4 \gets t_1 - t_3$\;
\If{$t_4 >= q$}{
    $z \gets t_4 - q$\;
}\Else {
    $z \gets t_4$\;
}
\Return{$z$}\;
\end{algorithm}

Inspired by Barrett Reduction~\cite{barrett1986implementing}, Barrett modular multiplication is a method that efficiently calculates the product of two integers within a specified modulus $q$. This approach, as presented in Algorithm~\ref{alg:SoftwareBarrettModMul}, utilizes multiplicative and bit shift operations to circumvent the slower process of division, thereby improving computational efficiency. This technique is particularly useful in cryptography, where fixed moduli are prevalent, and is essential for ensuring both security and efficiency in modular arithmetic. While Algorithm~\ref{alg:SoftwareBarrettModMul} effectively removes the need for division by precomputing the $k$ and $m$ and storing them in the register. The multiplication in line 4 is a bottleneck when implementing in DSPs. If both variables $a$ and $b$ are 32-bit integers, the resulting width of $t_1 \cdot m$ would be 96 bits. In Section~\ref{subsec:multiplication}, we will present our hardware-friendly implementation of Barrett Reduction in Algorithm~\ref{alg:HardwareBarrettModMul}.

\subsection{Residue Number System}
In FHE applications, adjusting the width of the coefficients in the polynomials is crucial for managing the trade-offs between noise tolerance, security, and computational efficiency. For example, a 4096-degree polynomial requires coefficients with a bit width greater than 109 bits~\cite{HomomorphicEncryptionSecurityStandard}. As a result, the computations in FHE require significant computing resources to handle the large data sizes.

RNS is a mathematical technique that decomposes large bit-width polynomials into multiple smaller bit-width polynomials, making them more manageable and adaptable to hardware platforms. By decomposing a polynomial in the $\mathbb{Z}_{Q}$ ring into $N_\text{q}$ polynomials in the rings $\mathbb{Z}{q_1},\mathbb{Z}{q_2},\dots,\mathbb{Z}{q_{N_\text{q}}}$, where $q_1,q_2,q_3\dots q_{N_\text{q}}$ are coprime and $Q=q_1q_2q_3\dots q_{N_\text{q}}$, RNS enables the transformation of polynomial coefficients into hardware-friendly representations. This approach effectively minimizes resource overhead in 32, 48, or 64-bit FPGA implementations, leading to improved performance and optimized resource utilization.

\section{Hazard-Free Number Theoretic Transform}
\label{sec:method}

While NTT expedites polynomial multiplications, there are both spatial and temporal dependencies in the process. Effectively managing these dependencies is essential for maximizing the efficiency of NTT-based polynomial multiplication. In this section, we propose a data layout and memory access pattern to eliminate all the hazards introduced by these dependencies in our accelerator.

\subsection{Spatial and Temporal Dependencies}
The NTT process for a polynomial of degree $N$ comprises $\log N$ stages, requiring $N/2 \cdot \log N$ butterfly operations. Figure~\ref{fig: Dataflow of NTT Operation for 16 Polynomial Degree} shows the process of both NTT with the CT algorithm and INTT with the GS algorithm for a polynomial of degree $16$. There are four stages in both NTT and INTT, with each stage routes differently indexed coefficients to the processing element (PE).
\begin{enumerate}
    \item Spatial dependencies. In all stages, for any PE to compute without delays, its two operands have to come from different memory banks. For instance, with $N = 16$, $a[0]$ is paired with $a[8], a[4], a[2]$, and $a[1]$ in different stages. To achieve zero hazards from the spatial dependencies, $a[0]$ has to reside in a different memory bank with all the paired elements.
    \item Temporal dependencies. Significant data dependencies exist between two adjacent stages, with each subsequent stage relying on results from its predecessor. As a result, the Read-after-Write (RAW) hazard could become a bottleneck for the accelerator, potentially leading to performance-degrading stalls.
\end{enumerate}

\subsection{Data Layout}
\begin{figure}[h]
    \centering
    \includegraphics[width=1\columnwidth]{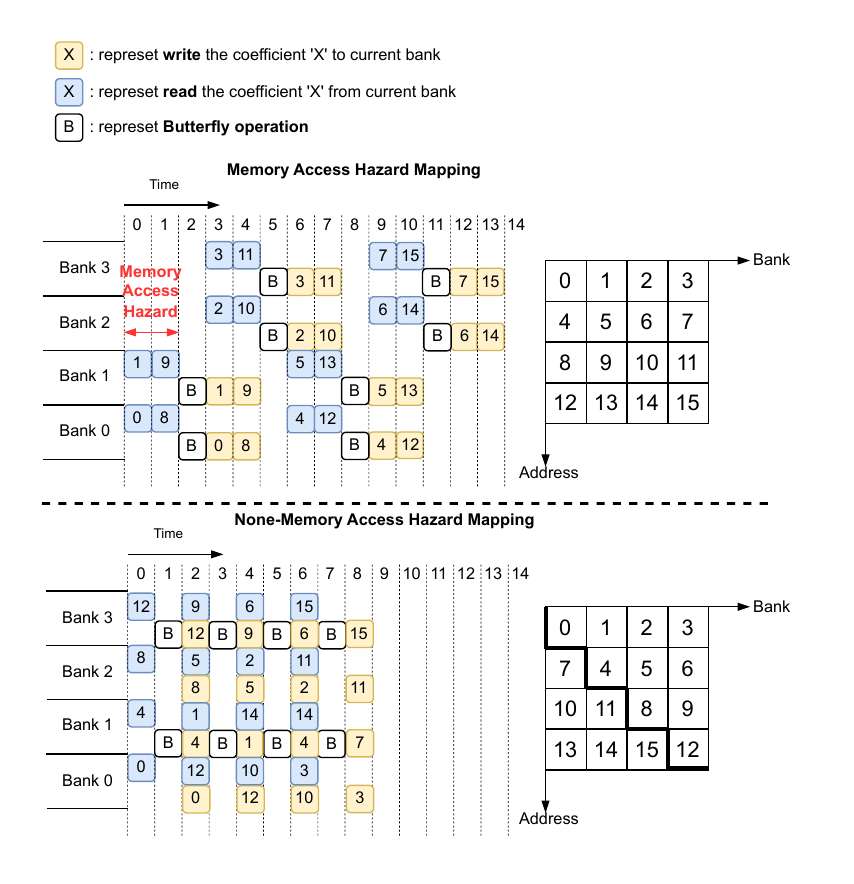}
    \caption{Example of Memory Access Hazard}
    \label{fig: Example of Memrory Access Hazard}
\end{figure}

In this section, we propose a data layout for coefficients to address the potential hazard introduced by spatial dependencies. Instead of storing all the coefficients in sequential order, we propose a method where coefficient are shifted one position to the right for each consecutive row, with the example of $N = 16$ shown in Figure~\ref{fig: Example of Memrory Access Hazard}. As we can see, sequentially storing the coefficients can lead to memory access hazard as $a[0]$ and $a[8]$ are in the same memory bank. By contrast, our proposed data layout, as demonstrated in Figure~\ref{fig: Example of Memrory Access Hazard} with $N = 16$, eliminates such hazards. We will further demonstrate the efficacy of this arrangement for any degree $N$.

Formally, the memory position of $a[i]$ is given by Equation~\ref{eq:storage_location}, where the $Bank$ number is shifted right by $Address$ positions compared to sequential storage. 
\begin{equation}
(Address, Bank) = \left(\left\lfloor \frac{i}{n} \right\rfloor, (i \bmod n + \left\lfloor \frac{i}{n} \right\rfloor) \bmod n\right)
\label{eq:storage_location}
\end{equation}
where $n = \sqrt{N}$.

Next, we will prove that $a[i]$ is in a different bank with all coefficients in $\{a[i \pm 1], a[i \pm 2], a[i \pm 4], ..., a[i \pm 2^t], ..., a[i \pm N/2]\}$, which is a sufficient condition to guarantee no memory access hazard in the butterfly operations.

\begin{theorem}
    For the data layout in Equation~\ref{eq:storage_location}, $a[i]$ is in a different bank with all coefficients in $\{a[i + 1], a[i + 2], a[i + 4], ..., a[i + 2^t], ..., a[i + N/2]\}$ for $N > 4~(n > 2)$. 
\end{theorem}

\begin{proof}
    Let $N = 2^k$, then $n = 2^{k/2}$. Consider the banks for $a[i]$ and $a[i + 2^t]$, where $0 \leq t \leq k - 1$.

    $Bank[i] = (i \bmod n + \left\lfloor \frac{i}{n} \right\rfloor) \bmod n)$

    $Bank[i + 2^t] = ((i + 2^t) \bmod n + \left\lfloor \frac{(i + 2^t)}{n} \right\rfloor) \bmod n)$

    \textbf{Case 1:} If $2^t >= n$, then $2^t \bmod n = 0$ and $2^t / n = 2^{t - k/2}$
    \begin{align*}
        Bank[i + 2^t] &= ((i + 2^t) \bmod n + \left\lfloor \frac{(i + 2^t)}{n} \right\rfloor) \bmod n\\
        &= (i \bmod n + \left\lfloor \frac{i}{n} \right\rfloor + 2^{t - k/2}) \bmod n\\
        &= (Bank[i] + 2^{t - k/2}) \bmod n
    \end{align*}
    
    As $1 \leq 2^{t - k/2} < n$, $Bank[i + 2^t] \neq Bank[i]$.

    \textbf{Case 2:} If $2^t < n$, i.e., $1 \leq 2^t \leq n / 2$, then $2^t \bmod n = 2^t$
    \begin{align*}
        &Bank[i + 2^t] = ((i + 2^t) \bmod n + \left\lfloor \frac{(i + 2^t)}{n} \right\rfloor) \bmod n\\
        &= (i \bmod n + 2^t + \left\lfloor \frac{i}{n} + \frac{1}{2^{k/2 - t}} \right\rfloor) \bmod n\\
        &= \begin{cases}
            (Bank[i] + 2^{t}) \bmod n, &\text{if} \left\lfloor \frac{i}{n} + \frac{1}{2^{k/2 - t}} \right\rfloor = \left\lfloor \frac{(i + 2^t)}{n} \right\rfloor\\
            (Bank[i] + 2^{t} + 1) \bmod n, &\text{otherwise}
        \end{cases}
    \end{align*}

    As $1 \leq 2^t \leq n / 2$ and $n > 2$, $Bank[i + 2^t] \neq Bank[i]$.

    Hence, we have shown that $a[i]$ and $a[i + 2^t]$ cannot reside in the same bank, for $1 \leq 2^t \leq N/2$.
\end{proof}

Similarly, we can prove that $a[i]$ is in a different bank with all coefficients in $\{a[i - 1], a[i - 2], a[i - 4], ..., a[i - 2^t], ..., a[i - N/2]\}$. As a result, we have proved that our data layout can eliminate all the hazards introduced by the spatial dependency.

\begin{figure*}[t]
    \centering
    \includegraphics[width=1\textwidth]{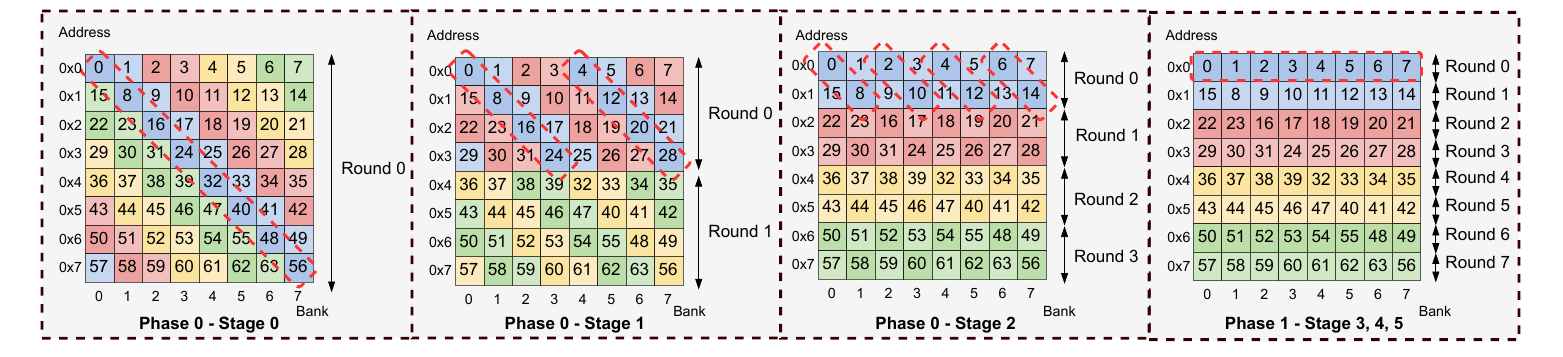}
    \caption{Dataflow Scheduling Strategy for $N$ = 64, $N_{pe}$ = 4}
    \label{fig:Proposed Dataflow}
\end{figure*}

\subsection{NTT/INTT Dataflow Scheduling}

The temporal dependencies between two adjacent stages may lead to RAW hazards, leading to performance-degrading stalls. In this section, we present our scheduling strategy for NTT and INTT to eliminate all RAW hazards. We denote the number of PEs as $N_{pe}$. To optimize the utilization of PEs is limited by the available memory banks, with $N_{pe} = 1, 2, 4, \ldots, n/2$. The likelihood of encountering RAW hazards increases with the number of PEs, due to their rapid data fetching before the data is computed in the previous stage. As a result, we will explain our dataflow scheduling strategy with $N_{pe} = n/2$.

The example of our dataflow scheduling strategy for $N = 64$ and $N_{pe} = 4$ is shown in Figure~\ref{fig:Proposed Dataflow}. We divide the computation processes for NTT and INTT into two distinct phases, each comprising $(\log N)/2$ stages. 
\begin{enumerate}
    \item In phase $0$, there are $(\log N)/2$ stages, with all data read diagonally. In Stage $s$, there are $2^s$ rounds, with each round requiring $n$ clock cycles to read all necessary data assuming $N_{pe} = n/2$. 
    The address range for the $r^{th}$ round of Stage $s$ is $[r\frac{N}{2^s}, (r + 1)\frac{N}{2^s})$.
    \item In phase $1$, there are also $(\log N)/2$ stages, where data is read sequentially in a row-by-row manner at each stage.
\end{enumerate}
After computing in the PE, all data are written back to their original locations. 

The comparison of our dataflow scheduling strategy with different computation speeds of PEs is shown in Figure~\ref{fig:Read_after_Write_Hazard}. As shown in the top part, when the computation speed of PE is slow, many RAW hazards can occur in this dataflow. For example, the computation result of $a[4]$ by PE0 in cycle 4 of the first stage of Phase 0 is needed in cycle 8, which is the first cycle of the second stage. Hence, the total time for PE computation and read/write delays to memory must be limited to 4 cycles. To prevent RAW stalls between Stage 0 and Stage 1, the computation delay and $n$ must satisfy the condition $delay_{compute} + delay_{io} < n/2$. In general, the requirement can be expressed as $delay_{compute} + delay_{io} < n(1 - 1/2^{k+1})$ to prevent any the RAW stalls between Stage $k$ and Stage $k+1$ in Phase 0. In Phase 1, as the data is read row-by-row, the requirement becomes $n > delay_{compute} + delay_{io}$. Therefore, the strictest condition to eliminate all RAW hazards is expressed in Equation~\ref{eq:RAW_constraint1}.

\begin{equation}
delay_{pe} + delay_{write} + delay_{read} < \frac{n}{2}
\label{eq:RAW_constraint1}
\end{equation}
When considering the case $N_{pe} < n/2$, each diagonal takes $\frac{n}{2N_{pe}}$ cycles. Then, the requirement becomes:
\begin{equation}
delay_{pe} + delay_{write} + delay_{read} < \frac{n}{2} \times \frac{n}{2N_{pe}}
\label{eq:RAW_constraint2}
\end{equation}
In practical applications of FHE, the degree $N$ is commonly set to values at least $2^{10}$ ($n \geq 32$). Based on Equation~\ref{eq:RAW_constraint1}, as long as the total time of PE computation and read/write delays is less than $16$ cycles, as is the case in our proposed design which will be introduced in the subsequent section, there would be no stalls due to RAW hazards.


\begin{figure}[h]
    \centering
    \includegraphics[width=1\columnwidth]{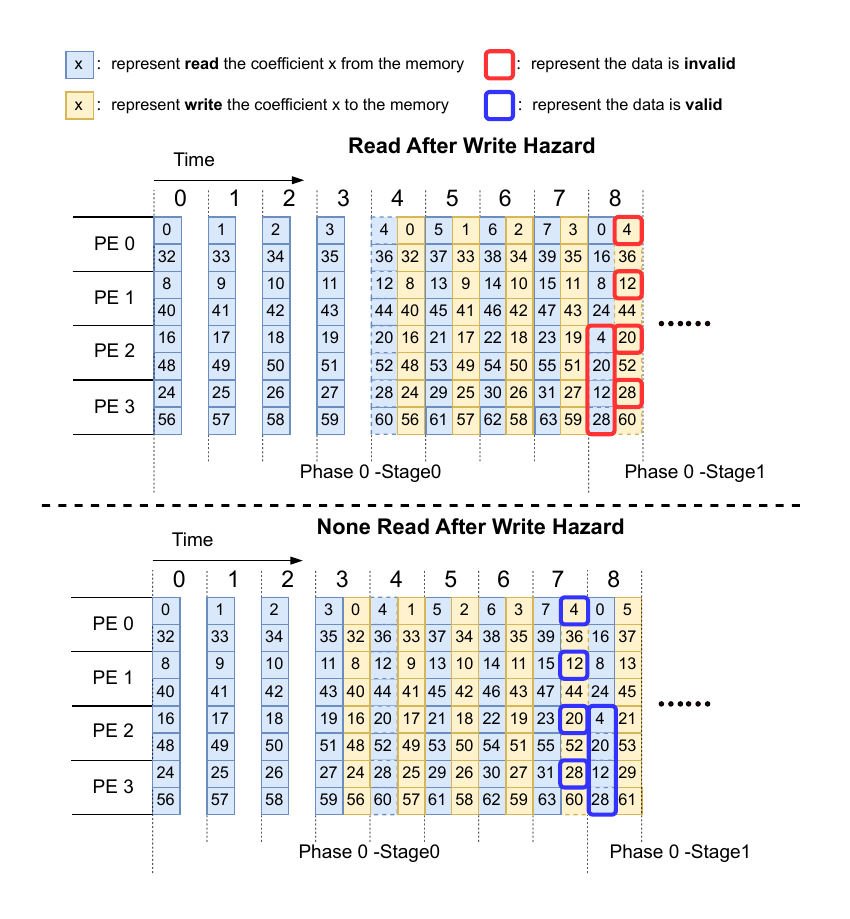}
    \caption{Example of Read After Write Hazard}
    \label{fig:Read_after_Write_Hazard}
\end{figure}
For INTT based on the GS algorithm, we only need to perform the calculation operations in the reverse order.

\section{Overall Microarchitecture}
\label{sec:arch}
\begin{figure*}[ht]
    \centering
    \includegraphics[width=0.8\linewidth]{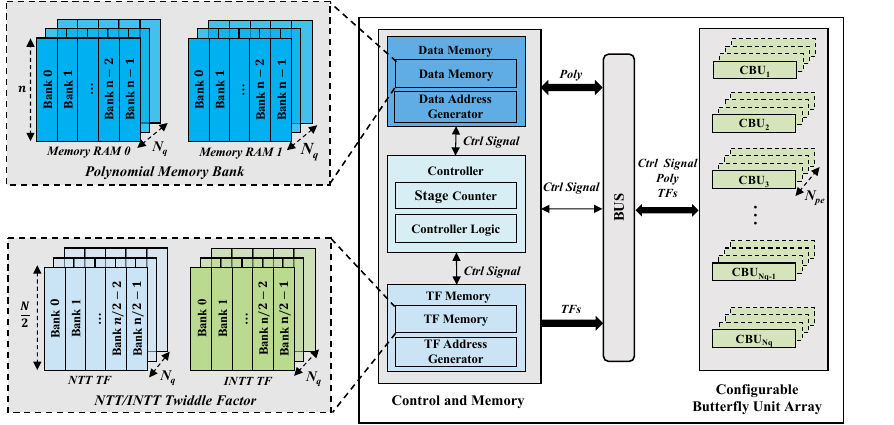}
    \caption{Overall Architecture}
    \label{Overall Architecture}
\end{figure*}
\subsection{Overall Architecture}
In the FHE scheme, a ciphertext consists of two sets of polynomials, each with $N$ coefficients adhering to modulus $Q$, where $N$ is a power of two. RNS decomposition is applied to decompose the polynomial into $N_\text{q}$ distinct polynomials with $N$ degrees, each associated with a specific modulus $q_i$, where $i$ ranges from $1$ to $N_\text{q}$.

The overall architecture is illustrated in Fig.~\ref{Overall Architecture}, consisting primarily of four essential modules: a controller module, a memory module, a computation module, and a bus module for communication. 

The controller module operates as a state machine, producing control signals in response to input stimuli. These signals manage data memory and twiddle factor memory by generating suitable read and write addresses. Furthermore, the control logic produces distinct control signals depending on the current system state, allowing the configurable butterfly unit to establish diverse data paths and perform various computational tasks. 

The data memory module consists of an address generation module and $2 \times N_\text{q}$ dual-ports block RAM arrays, with each RAM array containing $n$ memory banks of depth $n$ that stores $N$ coefficients. The address generation module generates data read and write addresses based on signals from the control logic. For the number theoretic transform, it needs completely different twiddle factors for NTT and INTT operation because CT and GS algorithms, which requires the twiddle factor memory to store twiddle factors for all modulus $q_i$, totaling $n/2 \times N \times N_\text{q}$ twiddle factors. A single twiddle memory RAM comprises $n/2$ memory banks, each with a depth of $N/2$, utilized to store the entire twiddle factors for NTT or INTT. Once written, these factors remain unchanged, and subsequent stages read them according to control signals. The bus ensures the transfer of control signals and data among various modules.

The computational module consists of $N_q $ configurable butterfly unit arrays, with each modulus $q_i$ associated with a CBU array. Each CBU array contains $N_\text{pe}$ configurable butterfly units, providing flexibility for configuring diverse data paths and adaptability for tasks such as NTT, modular multiplication, and INTT. A detailed explanation of these configurable processing units will be discussed in the following section.

\subsection{Configurable Butterfly Unit}
In polynomial multiplication based on NTT, our approach involves the integration of NTT, modular multiplication (MultMod), and INTT operations within a streamlined and compact CBU. We integrate the CT algorithm for NTT and GS algorithm for INTT operations, as shown in Fig.~\ref{fig:Configurable Butterfly Unit}(1). Unlike previous approaches with separate implementations for each operation, our CBU allows for an adaptable configuration of data paths for these operations, providing a more versatile solution. Illustrated in Fig.~\ref{fig:Configurable Butterfly Unit}(4) to Fig.~\ref{fig:Configurable Butterfly Unit}(6), three different modes of operations are supported through the MUX path adjusted by the control signal.

In the NTT mode, as depicted in Fig.~\ref{fig:Configurable Butterfly Unit}(4), the CBU is set up to execute the radix-2 CT algorithm. This configuration requires two integer inputs $a$ and $b$, a modulus $q$, and a twiddle factor $\omega$, which varies according to the specific stage and cycle of the computation. The outputs from the CBU in this mode are given by the expressions $a+b\omega$ and $a-b\omega$. In the INTT mode, shown in Fig.~\ref{fig:Configurable Butterfly Unit}(5), the CBU is reconfigured to support the radix-2 GS algorithm. This mode also takes two integers $a$ and $b$, an inverse twiddle factor $\omega$, and the modulus $q$. The outputs are respectively $(a+b)/2$ and $(a-b)/2$ The third configuration is the MultMod mode, illustrated in Fig.~\ref{fig:Configurable Butterfly Unit}(6), where the CBU is performing modular multiplication operations using the Barrett Multiplication algorithm. In this mode, the inputs are two integers, $a$ and $b$, together with the modulus $q$. The output is the result of the modular multiplication $(a \cdot b) \mod q$.

\begin{figure*}[h]
    \centering
    \includegraphics[width=0.98\textwidth]{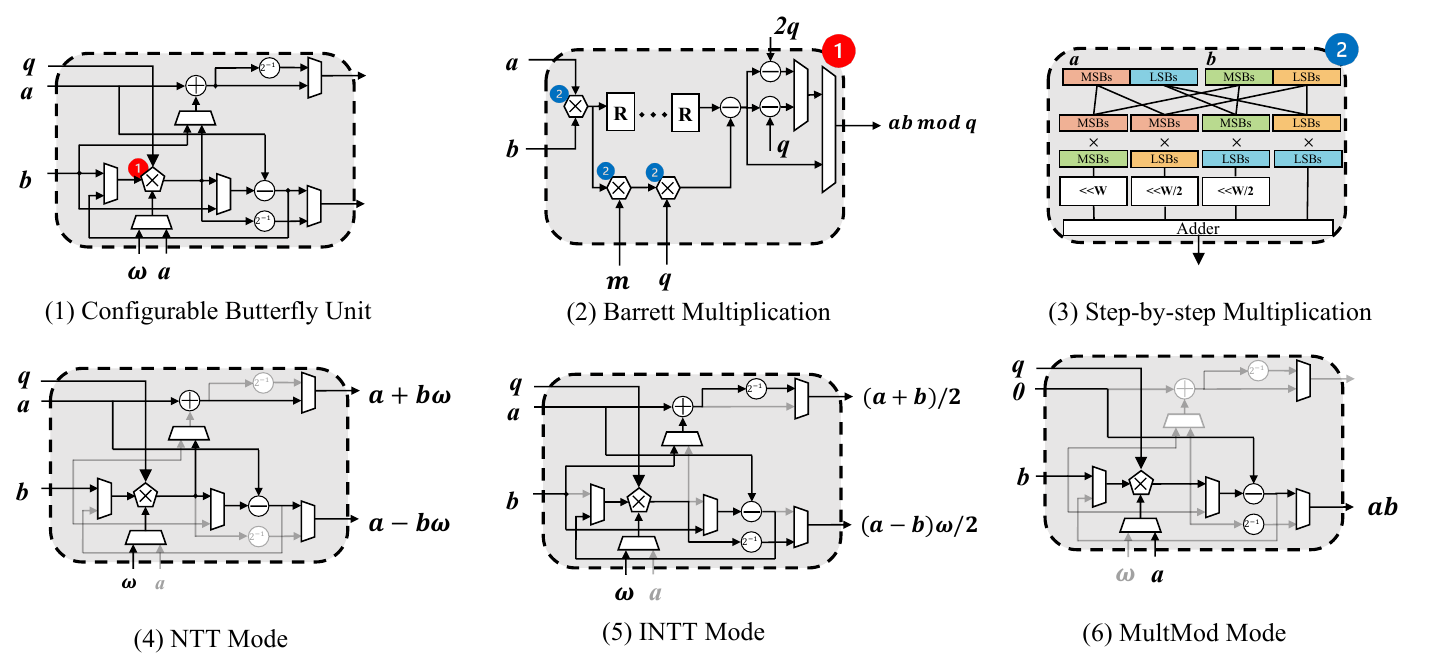}
    \caption{(1) Configurable Butterfly Unit (The detailed structure of the multiplier \ding{182} is in (2)). (2) Barrett Multiplication (The detailed structure of multipler \ding{183} is in (3)). (3) Step-by-step Multiplication. (4)-(6) CBU in different modes. }
    \label{fig:Configurable Butterfly Unit}
\end{figure*}

This unified data flow control consolidates control mechanisms for three modes and allows for efficient implementation of arithmetic operations required in the FHE scheme, thereby reducing the overall hardware complexity and improving resource utilization.

\subsection{Efficient Modular Multiplication}
\label{subsec:multiplication}
Among the supported operations, including addition modulo, subtraction modulo, and modular multiplication, the last one requires special attention due to its higher resource consumption and optimization necessity for achieving acceptable resource utilization and performance. Two prevalent modular multiplication algorithms, Montgomery and Barrett, were examined in terms of hardware resource utilization and data bit width. Results show that the Montgomery algorithm provides superior resource efficiency with very large bits, but requiring data preprocessing and extensive processing time. In contrast, the Barrett algorithm has advantages with relatively small bit widths and can leverage Digital Signal Processors (DSPs) to significantly enhance the performance. 

Given that the coefficient width is maintained below 32 bits following RNS decomposition, the Barrett algorithm was chosen for performing modular multiplication within the CBU, with its hardware-specific implementation depicted in Algorithm~\ref{alg:HardwareBarrettModMul}. The overall architectural design of this implementation within the CBU is illustrated in Fig.~\ref{fig:Configurable Butterfly Unit}(2).

\begin{algorithm}[ht]
\caption{Hardware-friendly Barrett Modular Multiplication}
\label{alg:HardwareBarrettModMul}
\KwIn{$a, b, q \in \mathbb{Z}^+$, $a, b < q$}
\KwOut{$z = (a \cdot b) \pmod{q}$}

\BlankLine
$k \gets \lceil \log_2(\textit{q}) \rceil$\;
\BlankLine
$m \gets \lfloor 2^{2k} / \textit{q} \rfloor$\;

\BlankLine
\textit{$t_1$} $\gets$ \textbf{StepMultiplier}(\textit{a}, \textit{b})\;

\BlankLine
\textit{$t_1^{high}$} $\gets$ \textit{$t_1$} $>> (k - 1)$\;

\BlankLine
\textit{$t_2$} $\gets$ \textbf{StepMultiplier}(\textit{$t_1^{high}$}, $m$) $>> (k + 1)$\;

\BlankLine
\textit{$t_3$} $\gets$ \textbf{StepMultiplier}(\textit{$t_2$}, \textit{q})\;

\textit{$t_4$} $\gets$ \textit{$t_1$} - \textit{$t_3$}\;

\If{$t_4 >= 2q$}{
    $z \gets t_4 - 2q$\;
}\ElseIf{$t_4 >= q$}{
    $z \gets t_4 - q$\;
}\Else{
    $z \gets t_4$\;
}

\Return{$z$}\;
\end{algorithm}
In order to prove the correctness of Algorithm~\ref{alg:HardwareBarrettModMul}, we first examine the relationship between $t_2$ in Algorithm~\ref{alg:SoftwareBarrettModMul} and $t_2$ in Algorithm~\ref{alg:HardwareBarrettModMul}, followed by a proof of the equivalence of these two algorithms.

\begin{lemma}
    $m$ in Algorithm~\ref{alg:HardwareBarrettModMul} is $k + 1$ bits.
\end{lemma}
\begin{proof}
\begin{equation*}
\begin{aligned}
    k = \lceil \log_2(\textit{q}) \rceil
    \implies & k - 1 < \log_2(\textit{q}) \leq k \\
    \implies & 2^{k-1} < q \leq 2^k \\
    \implies & 2^k \leq m = \lfloor 2^{2k} / \textit{q} \rfloor < 2^{k+1}
\end{aligned}
\end{equation*}
\end{proof}

\begin{theorem}
\label{thm:t2}
The value of $t_2$ in Algorithm~\ref{alg:HardwareBarrettModMul} is either equal to or one less than the value of $t_2$ in Algorithm~\ref{alg:SoftwareBarrettModMul}.
\begin{equation*}
    \begin{aligned}
    &t_2^{alg3} = t_2^{alg2} - 1, \text{or}\\
    &t_2^{alg3} = t_2^{alg2}
\end{aligned}
\end{equation*}
\end{theorem}
\begin{proof}
The computations of $t_1$ in Algorithm~\ref{alg:SoftwareBarrettModMul} and Algorithm~\ref{alg:HardwareBarrettModMul} are exactly the same. The difference is in the computation of $t_2$. While Algorithm~\ref{alg:SoftwareBarrettModMul} computes $a \cdot b \cdot m$ and then right shifts the result by $2k$ bits, Algorithm~\ref{alg:HardwareBarrettModMul} first right shifts the result of $a \cdot b$ by ($k - 1$) bits, subsequently multiplies it by $m$, and finally right shifts the result by $(k + 1)$ bits. Let us split $t_1$ into high and low parts, with the low part containing ($k - 1$) bits.
\begin{equation*}
    \begin{aligned}
    t_2^{alg2} &= (t_1^{high} \cdot 2^{k-1} + t_1^{low}) \cdot m >> (2k)\\
    &= (t_1^{high} \cdot m \cdot 2^{k-1} + t_1^{low} \cdot m) >> (2k)\\
\end{aligned}
\end{equation*}
As $m$ is $k + 1$ bits and $t_1^{low}$ is $k - 1$ bits, the right part $t_1^{low} \cdot m < 2^{2k}$. Therefore,
\begin{equation*}
    \begin{aligned}
  t_2^{alg2} &\leq ~((t_1^{high} \cdot m) >> (k+1)) + 1 = t_2^{alg3} + 1 \\
  t_2^{alg2} &\geq ~((t_1^{high} \cdot m) >> (k+1)) = t_2^{alg3}
\end{aligned}
\end{equation*}
\end{proof}

\begin{theorem}
\label{thm:t4}
The value of $t_4$ in Algorithm~\ref{alg:HardwareBarrettModMul} is either equal to or greater than the value of $t_4$ in Algorithm~\ref{alg:SoftwareBarrettModMul} by a quantity of q.
\begin{equation*}
    t_4^{alg3} = t_4^{alg2},~~~~ \textit{or}~~~~ t_4^{alg3} = t_4^{alg2} + q
\end{equation*}
\end{theorem}
\begin{proof}
\begin{equation*}
\begin{aligned}
    t_4^{alg2} &= t_1 - t_2^{alg2} \cdot q \\
    t_4^{alg3} &= t_1 - t_2^{alg3} \cdot q
\end{aligned}
\end{equation*}
Based on Theorem~\ref{thm:t2}, $t_4^{alg3} = t_4^{alg2}$, or $t_4^{alg3} = t_4^{alg2} + q$.
\end{proof}

The equivalence of Algorithm~\ref{alg:HardwareBarrettModMul} and Algorithm~\ref{alg:SoftwareBarrettModMul} can be easily demonstrated with the use of Theorem~\ref{thm:t4}.

Algorithm~\ref{alg:HardwareBarrettModMul} demonstrates greater hardware-friendliness compared to Algorithm~\ref{alg:SoftwareBarrettModMul} due to the right shift of the product of $a$ and $b$ before multiplication with $m$. In hardware implementation, optimizing multiplications is essential to align with the structural capabilities of FPGAs and the FHE scheme set. For the Virtex-7 family, the relevant DSP slice is the DSP48E1. In FHE algorithms, the coefficients typically exceed 30 bits, which cannot be directly handled by DSP48E1. To efficiently handle large bit-width data, it is necessary to partition the data into smaller segments, which involves dividing the data into two sections: the Most Significant Bits (MSBs) and the Least Significant Bits (LSBs), designated as $MSBs=[W:W/2]$ and $LSBs=[W/2-1:0]$ respectively, where $W$ represents the total data width. This segmentation is graphically represented in Fig.~\ref{fig:Configurable Butterfly Unit}(6). After the multiplication, the result can be recovered through shift and concatenation operations. It can potentially improve frequency using DSP and reduce the cascade of Look-Up Tables (LUTs) and Flip-Flops (FFs). In the context of small parameter designs, such as $log_2{q} = 14$, these multiplications can be efficiently processed by DSP without segmentation.

\subsection{Efficient Half Modular Operation} By employing the GS algorithm, we perform multiplication of the result by $1/2$ using a streamlined hardware configuration consisting of a shift register and an adder. This refined substitution eliminates the necessity for post-computation multiplication by $N^{-1}$, resulting in reduced bit-width operations and decreased hardware overhead throughout the computational process.

The operation $x/2 \mod q$ can be efficiently computed using Theorem~\ref{define:x/2modq}. 

\begin{theorem}\label{define:x/2modq}
For a given integer $x$ and modulus $q$, the expression $x/2 \mod q$ can be mathematically represented as:
\begin{equation}
x/2 \bmod q= 
\left\{\begin{matrix}
 (x>>1) \bmod q  & x=2k,k\in\mathbb{Z} \\
  \left\lfloor\frac{x}{2}\right\rfloor+\frac{q+1}{2}(\bmod q) & x=2k+1, k\in \mathbb{Z} 
\end{matrix}\right.
\end{equation}
\vspace{0.15in}
\end{theorem}

\section{Experiment}
\label{sec:exp}
We implemented our NTT accelerator HF-NTT using Chisel and synthesized it to Verilog. By adjusting parameters such as the RNS decomposition modulus number $N_q$, modulus size $q_i$, polynomial length $N$, and number of CBU $N_{\text{PE}}$, we can generate the required hardware specifications.

To ensure a fair comparison for evaluating our accelerator, we utilized the Xilinx Vivado 2022.2 toolchain for synthesis and implementation. The target platform is the VIRTEX-7 xc7vx1140tflg1926 FPGA, which features 1,139,200 LUTs, 1,424,000 Flip-Flops, 3,360 DSPs, and 67,680 Kb of BRAM.

\subsection{Clock Cycles}

Our first experiment was conducted to investigate the relationship between the number of CBUs and the clock cycles needed for NTT, modular multiplication, and INTT operations using the generated accelerator with the data length $N=4096$. The results are summarized in Table~\ref{the relationship between cc and npe}. As we can see, the clock cycles for our accelerator is very close to the ideal case, since we do not have any stalls. The only difference is a small number of cycles for the initial setup.

\begin{table}[h]
\centering
\renewcommand\arraystretch{1.2}
\caption{The relationship between clock cycles and $N_{pe}$  $(N = 4096)$}
\label{the relationship between cc and npe}
\begin{tabular}{cccccc}
\toprule
\multirow{2}{*}{$N_{pe}$} & \multicolumn{4}{c}{Clock cycles (CC)}   \\
\cmidrule{2-6}
& Ideal NTT & NTT~\cite{9882331} & NTT & Mult & INTT  \\
\midrule
1  & 24,576 & 24,583 & 24,595 & 4,114 & 24,596 \\
2  & 12,288 & 12,295 & 12,307 & 2,066 & 12,308 \\
4  & 6,144 & 6,151 & 6,163  & 1,042 & 6,164  \\
8  & 3,072 & 3,079 & 3,091  & 530   & 3,092  \\
16 & 1,536 & 1,543 & 1,555  & 274   & 1,556  \\
32 & 768 & 775 & 787    & 146   & 788    \\
\bottomrule
\end{tabular}
\end{table}

Indeed, the relationships between the required clock cycles and the number of CBUs ($N_{pe}$) can be expressed in the following equation:

\begin{equation}
    CC_{\text{(I)NTT}} = \frac{N \log N }{2N_{\text{pe}}} + CC_{\text{butterfly}} + CC_{\text{read}} + CC_{\text{write}}
\end{equation}

where $CC_{\text{butterfly}}$ for INTT is one clock cycle more than that for NTT. This results in $CC_{\text{INTT}}$ being one clock cycle longer than $CC_{\text{NTT}}$ in the entire hazard-free pipeline.

Compared with the work presented in~\cite{9882331}, our implementation requires more clock cycles to complete the tasks. This is because we have optimized the pipeline by adding more registers, allowing us to achieve a higher operating frequency than that reported in~\cite{9882331}. Despite the increased number of clock cycles, the overall time consumes less due to the higher operating frequency.

\subsection{Performance and Area Evaluation}
We conducted an analysis of the relationship between performance and hardware resource overhead under the specific condition of $N_\text{q} =1, \lceil\log_2q\rceil=32$, and $N=4096$.  Both the NTT and polynomial multiplication operations exhibit similar trends in latency reduction. 

It is important to note that the number of BRAMs remains constant regardless of changes in the number of processing units, as their memory requirements are not directly affected by the number of CBUs but determined by polynomial degree $N$, modulus size $q_i$ and modulus number $N_q$. However, the utilization of other resources, such as LUTs, FFs, and DSPs, increases with the number of CBUs. 

The decrease in clock cycles is particularly noteworthy. The latency is reduced by 50\% as the number of CBUs is doubled, underscoring the significant efficiency gains achieved through increased parallelism. The execution time of NTT operations and the number of CBUs ($N_{pe}$). The execution time decreases significantly from $\mathbf{92.43 \mu s}$ to $\mathbf{2.74 \mu s}$ as we increase $N_{pe}$ from 1 to 32.

\begin{figure}[h]
    \centering
    \includegraphics[width=1\columnwidth]{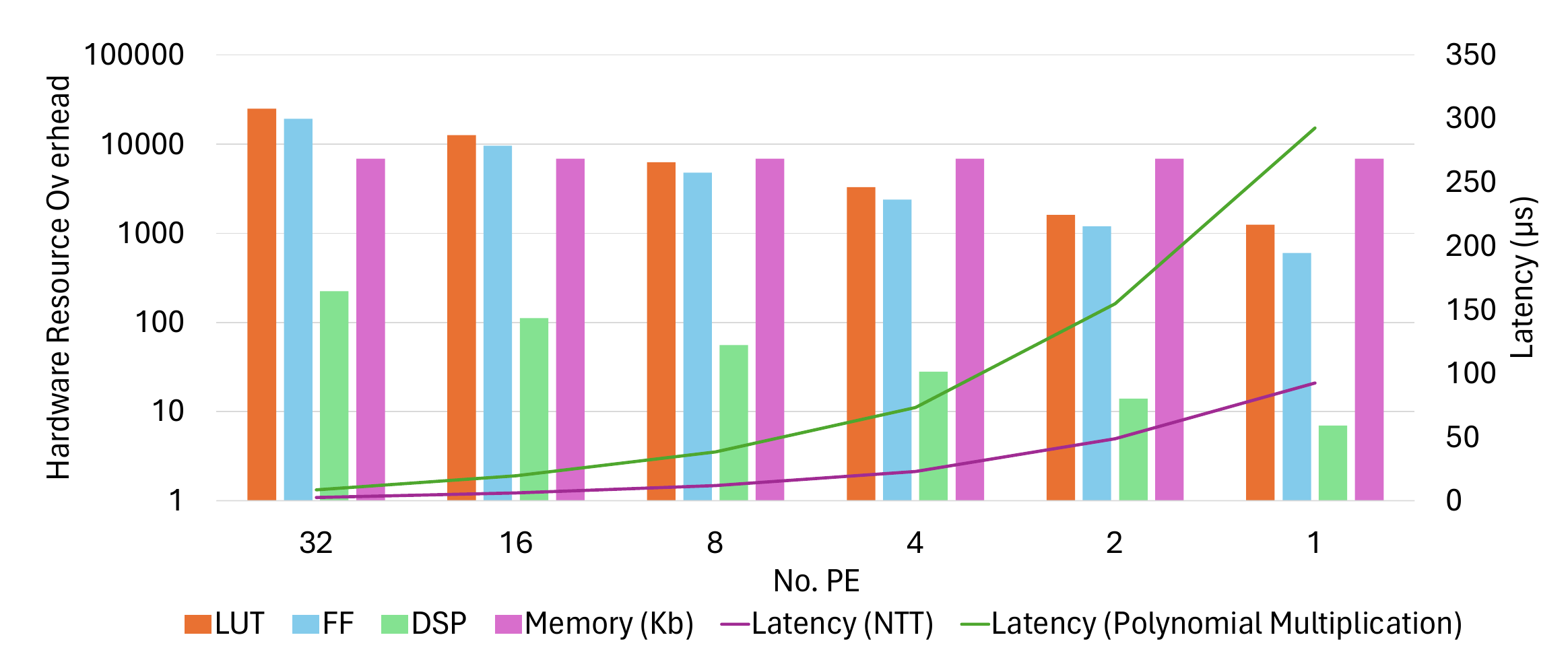}
    \caption{The relationship between performance, hardware resource and different CBU number $N_{pe}$  ($N_\text{q} =1, \lceil\log_2q\rceil=32$, and $N=4096$)}
    \label{fig:Latency and Resource}
\end{figure}

To evaluate our design efficiency, we employ the Area-Time Product (ATP) metric. The calculation involves summing LUT and FF resources, multiplied with the time consumed by the operation. Lower values of ATP indicate a better utilization of resources in achieving faster computational speeds.

The relationship between (LUT+FF)\_ATP and the CBU number $N_{\text{pe}}$ is illustrated in Figure~\ref{ATP}. When the polynomial degree $N = 1024$, the (LUT+FF)\_ATP does not change much as the PE number varies as the design is relatively simple for this small $N$. In all the other large cases, we observed a clear increasing trend as the number of CBUs decreases. This observation indicates that as parallelism increases, the benefits in execution time outweigh the increased resource overhead, resulting in the (LUT+FF)\_ATP decreases. It shows that our accelerator has superior efficiency for large-scale, highly parallel circuits, particularly in processing high-degree polynomials.



\begin{figure}[ht]
    \centering
    \begin{subfigure}[b]{0.49\columnwidth}
        \centering
        \includegraphics[width=\textwidth]{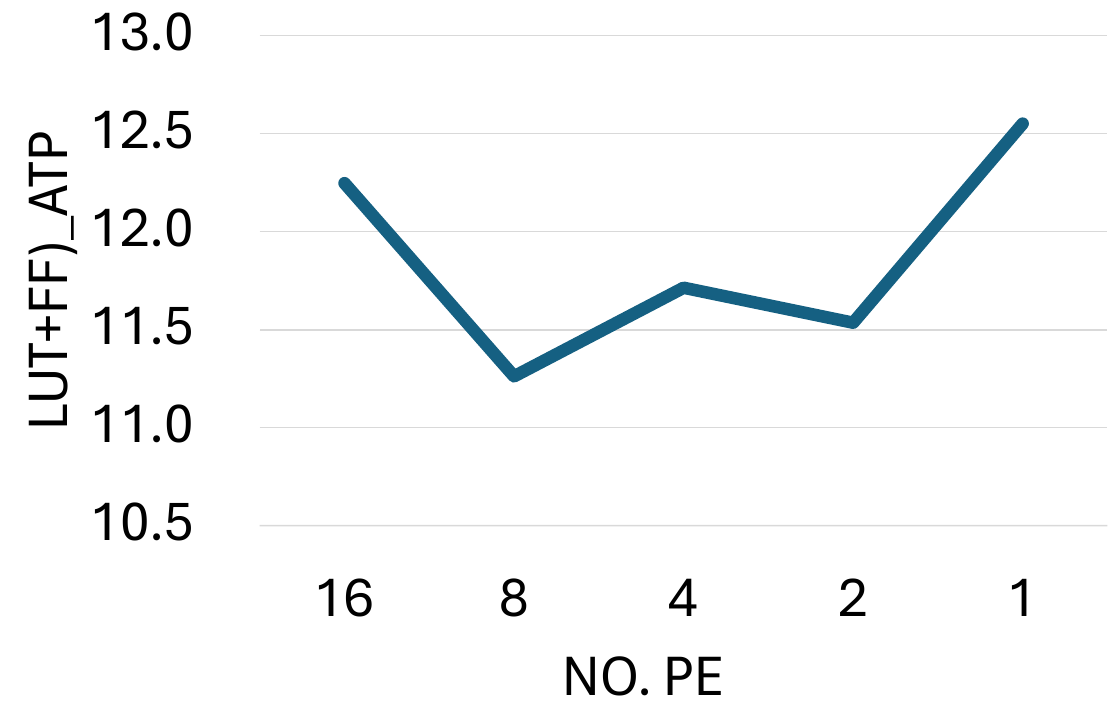}
        \caption{$N = 1024, \lceil \log_2q \rceil = 14$}
    \end{subfigure}
    \begin{subfigure}[b]{0.49\columnwidth}
        \centering
        \includegraphics[width=\textwidth]{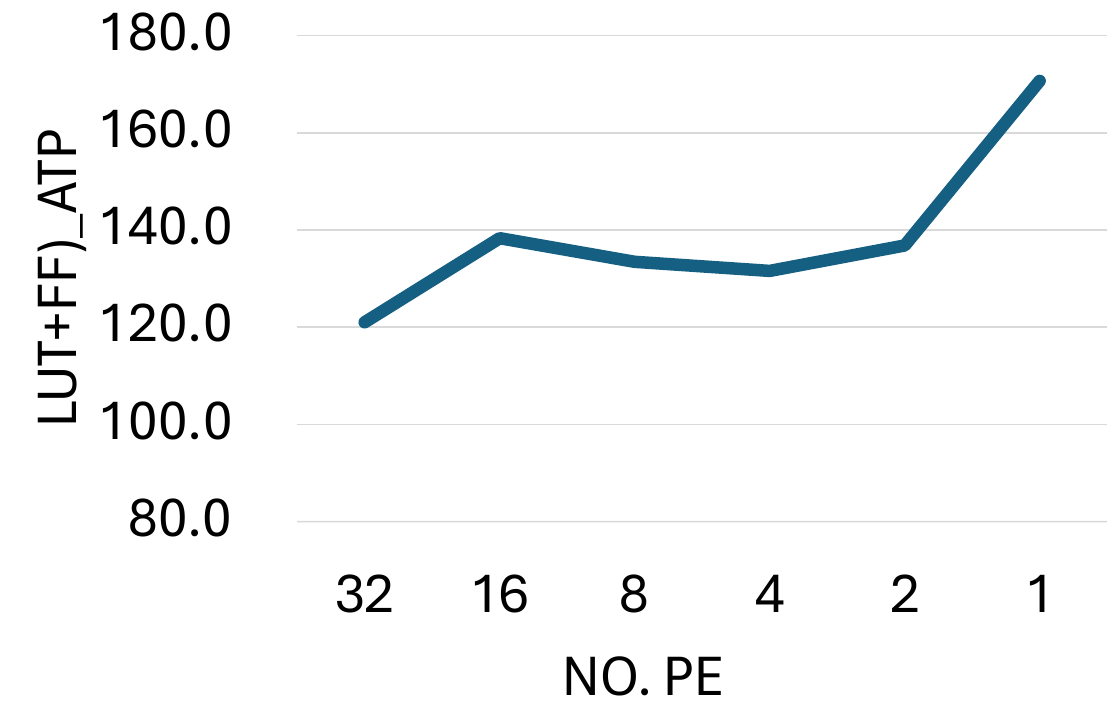}
        \caption{$N = 4096, \lceil \log_2q \rceil = 32$}
    \end{subfigure}
    
    \begin{subfigure}[b]{0.49\columnwidth}
        \centering
        \includegraphics[width=\textwidth]{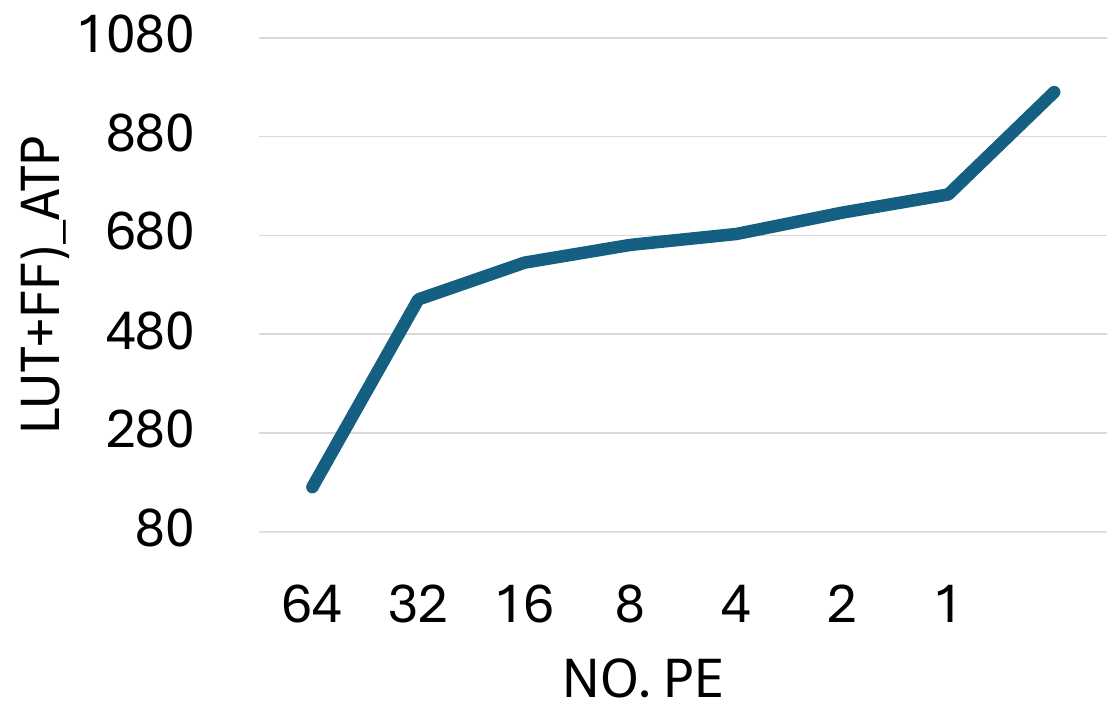}
        \caption{$N = 16384, \lceil \log_2q \rceil = 32$}
    \end{subfigure}
    \begin{subfigure}[b]{0.49\columnwidth}
        \centering
        \includegraphics[width=\textwidth]{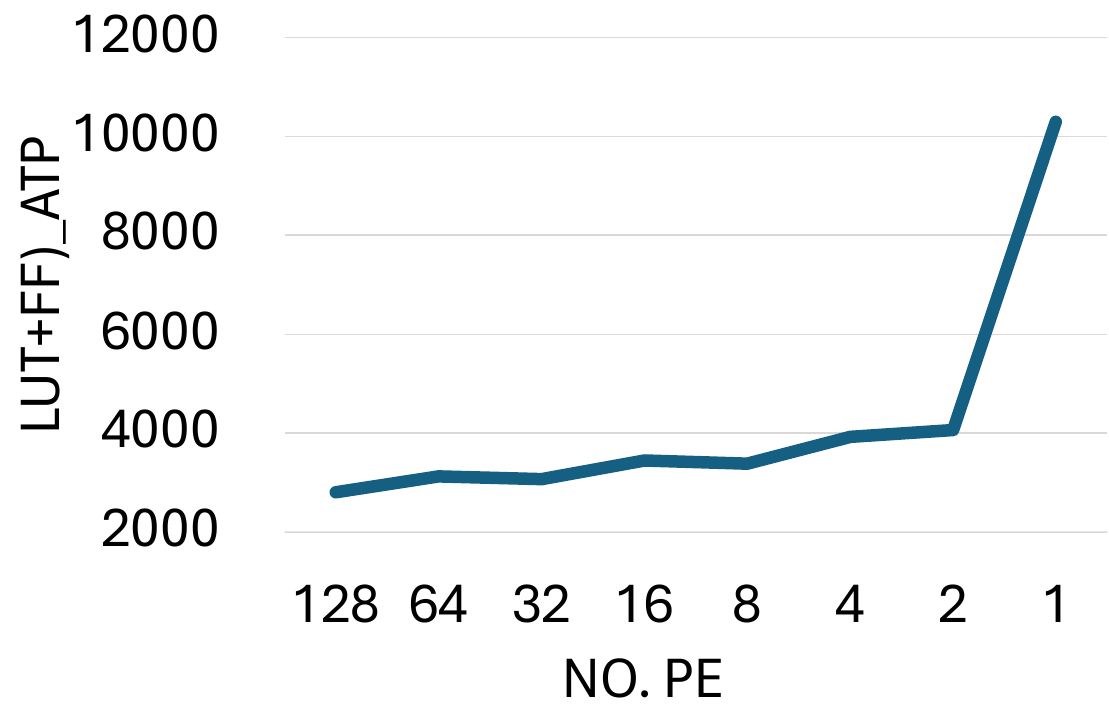}
        \caption{$N = 65536, \lceil \log_2q \rceil = 32$}
    \end{subfigure}  
    \caption{The relationship between (LUT+FF)\_ATP and CBU number $N_{pe}$ with different polynomial degree $N$}
    \label{ATP}
\end{figure}

\begin{table*}[h]
\caption{Evaluation results of our design under single-modulus and comparison with previous work}
\label{Evaluation results of our design under single-modulus and comparison with previous work}
\renewcommand\arraystretch{1.3}
\begin{tabular}{cccccccccc}
\hline
\rowcolor[HTML]{EFEFEF} 
Work & Platform & N & $log_2q$ & PE Num & LUT/FF/BRAM/DSP & Clock Cycles & Freq & NTT Latency & (LUT+FF)\_ATP       \\ 
\rowcolor[HTML]{EFEFEF} 
 &  &  &  &  &  & (NTT) & (Mhz) & ($\mu$s) &        \\ \hline
~\cite{zhang2020highly} & Artix-7 & & & 2 PE & 847/375/6/2     & 2569 & 244 & 10.5 & 12.8  \\ \cline{1-2} \cline{5-10} 
~\cite{kuo2017high}     & Artix-7 & & & -    & 2832/1381/10/8  & 1590 & 150 & 10.6 & 44.7 \\ \cline{1-2} \cline{5-10} 
                       &         & & & 1 PE & 575/-/11/3      & 5160 &     & 41.3 & 23.7 \\ \cline{5-7} \cline{9-10} 
                       &         & & & 8 PE & 2584/-/16/24    & 680  &     & 5.4  & 14.0 \\ \cline{5-7} \cline{9-10} 
\multirow{-3}{*}{~\cite{9171507}} & \multirow{-3}{*}{Virtex-7} & & & 32 PE & 17188/-/48/96 & 680 & \multirow{-3}{*}{125} & 1.6 & 27.5 \\ \cline{1-2} \cline{5-10} 
~\cite{9882331}         & Virtex-7& & & 16 PE& 8515/3618/12/49 & 334  & 172 & 1.9  & 23.1 \\ \cline{1-2} \cline{5-10} 
~\cite{10558123}       &Virtex-7 & & & - & 9100/7200/8/28 & 359 & 256 & 1.4 & 22.8\\  \cline{1-2} \cline{5-10} 
\textbf{Ours}          & Virtex-7& & & \textbf{1 PE} & \textbf{466/246/17.3\textsuperscript{a}/3}& \textbf{5135}& \textbf{291}& \textbf{17.6} & \textbf{12.6}  \\ \cline{1-2} \cline{5-10} 
\textbf{Ours}          & Virtex-7& & & \textbf{2 PE} & \textbf{820/486/17.3\textsuperscript{a}/6}& \textbf{2575}& \textbf{292}& \textbf{8.8} & \textbf{11.5}  \\ \cline{1-2} \cline{5-10} 
\textbf{Ours}          & Virtex-7& & & \textbf{4 PE} & \textbf{1738/967/17.3\textsuperscript{a}/12}& \textbf{1295}& \textbf{299}& \textbf{4.3} & \textbf{11.7}  \\ \cline{1-2} \cline{5-10} 
\textbf{Ours}          & Virtex-7& & & \textbf{8 PE} & \textbf{3267/1932/17.3\textsuperscript{a}/24}& \textbf{655}& \textbf{302}& \textbf{2.2} & \textbf{11.3}  \\ \cline{1-2} \cline{5-10} 
\textbf{Ours}          & Virtex-7& \multirow{-12}{*}{1024} & \multirow{-12}{*}{14} & \textbf{16 PE} & \textbf{6560/3859/17.3\textsuperscript{a}/48}& \textbf{335}& \textbf{285}& \textbf{1.2} & \textbf{12.2}  \\ \hline \hline
                       &         & &32& 8 PE & 6260/7132/20/56 & 3255 & 235 & 13.9 & 185.5 \\ \cline{4-10} 
                       
\multirow{-2}{*}{~\cite{10171572}} & \multirow{-2}{*}{Virtex-7} & &32 & 32 PE & 46547/26937/192/448 & 572& 184 & 3.1 & 228.5  \\ \cline{1-2} \cline{4-10} 
~\cite{8866755}       & Virtex-7  & & 32 & - & 70000/-/129/448    & 460 & 200 & 2.3 & 161.0 \\ \cline{1-2} \cline{4-10} 
~\cite{9761841}       & Virtex-7  & & 32 & - & 14004/8662/87/80   & 9750& 250 & 39.0& 884.0 \\ \cline{1-2} \cline{4-10} 
~\cite{9882331}       & Virtex-7  & & 24 & 16 PE & 14591/6468/12/80   & 1543& 121 & 12.8& 269.6 \\ \cline{1-2} \cline{4-10} 
~\cite{8675244}       & Zynq UltraScale+  & & 30 & - & 64000/-/400/200 &- & 225 & 73 & 4672\\ \cline{1-2} \cline{4-10} 
~\cite{10558123}       &Virtex-7 & & 32 & - & 23000/20800/8/128 & 1176 & 196 & 6 & 262.8\\  \cline{1-2} \cline{4-10} 

\textbf{Ours}        & Virtex-7 &  & 32 & \textbf{1 PE} & \textbf{1245/601/192/7} & \textbf{24595}& \textbf{266}  & \textbf{92.4} & \textbf{170.6} \\ \cline{1-2} \cline{4-10} 
\textbf{Ours}        & Virtex-7 &  & 32 & \textbf{2 PE} & \textbf{1604/1198/192/14} & \textbf{12307}& \textbf{252}  & \textbf{48.8} & \textbf{136.8} \\ \cline{1-2} \cline{4-10} 
\textbf{Ours}        & Virtex-7 &  & 32 & \textbf{4 PE} & \textbf{3301/2393/192/28} & \textbf{6163}& \textbf{267}  & \textbf{23.1} & \textbf{131.6} \\ \cline{1-2} \cline{4-10} 
\textbf{Ours}        & Virtex-7 &  & 32 & \textbf{8 PE} & \textbf{6247/4784/192/56} & \textbf{3091}& \textbf{255}  & \textbf{12.1} & \textbf{133.5} \\ \cline{1-2} \cline{4-10} 
\textbf{Ours}        & Virtex-7 &  & 32 & \textbf{16 PE} & \textbf{12671/9564/192/112} & \textbf{1555}& \textbf{250}  & \textbf{6.2} & \textbf{138.3} \\ \cline{1-2} \cline{4-10}

\textbf{Ours}        & Virtex-7 & \multirow{-13}{*}{4096} & 32 & \textbf{32 PE} & \textbf{25040/19203/192/224} & \textbf{787}& \textbf{288}  & \textbf{2.7} & \textbf{121.0} \\ \hline \hline

\cite{10.1145/3373376.3378523} & Arria10 GX 1150 &  & 27 & - & 582148/1554005/3986/2018 & 1536& 300  & 5.12 & 10936.1 \\ \cline{1-2} \cline{4-10} 
\cite{8318681} & XC6VLX240T &  & 30 & - & 72163/63086/53.4/2736 & 48000 & $100 \& 200$  & 480 & 64919.5 \\ \cline{1-2} \cline{4-10} 
\cite{10558123}       &Virtex-7 & & 32 & - & 26900/26900/32.5/144 & 4320 & 200 & 21.6 & 1162.1\\  \cline{1-2} \cline{4-10} 
\cite{8928952}      &Virtex-7 & & 32 & - & 2800/1200/80/39 & 4320 & 168 & 975.2 & 3900.8\\  \cline{1-2} \cline{4-10} 
\textbf{Ours}        & Virtex-7 &  & 32 & \textbf{16 PE} & \textbf{13826/9566/1152/112} & \textbf{7187}& \textbf{254}  & \textbf{28.3} & \textbf{660.9} \\ \cline{1-2} \cline{4-10} 
\textbf{Ours}        & Virtex-7 &  & 32 & \textbf{32 PE} & \textbf{25052/19142/1152/224} & \textbf{3603}& \textbf{255}  & \textbf{14.1} & \textbf{624.7} \\ \cline{1-2} \cline{4-10} 

\textbf{Ours}        & Virtex-7 & \multirow{-7}{*}{16384}  & 32 & \textbf{64 PE} & \textbf{49987/38151/1152/448} & \textbf{1811}& \textbf{290}  & \textbf{6.2} & \textbf{550.5} \\ \hline\hline


\end{tabular}
\smallskip
\textit{
    \textsuperscript{a}This part is synthesized as LUTRAMs (results transformed to RAMB36).}
\end{table*}

\begin{table*}[ht]
\centering
\renewcommand\arraystretch{1.3}
\caption{Evaluation results of our design under multi-modulus and comparison with previous work}
\label{Evaluation results of our design under large parameters and comparison with previous work}
\begin{tabular}{cccccccccc}
\hline
\rowcolor[HTML]{EFEFEF} 
Work & Platform & N & $log_2q$ & LUT/FF/BRAM/DSP & Freq (Mhz) & Latency ($\mu$s) & (LUT+FF)\_ATP       \\ \hline
~\cite{10171572} & Virtix-7 & 4096 & 60 & 43696/50322/192/448 & 204 & 4.68 & 440.0  \\ \hline
~\cite{9171507}  & Virtix-7 & 4096 & 60 & 99384/-/176/992     & 125 & 7.7  & 765.25  \\ \hline
~\cite{9116470} (Perf-Opt.)  & Virtix-7 & 4096 & 60 & 338000/138000/1984/768 & 125 & 7 & 3332  \\ \hline
~\cite{9116470} (Balanced)   & Virtix-7 & 4096 & 60 & 22000/17000/248/96 & 125 & 26 & 1014  \\ \hline
~\cite{9116470} (Area-Opt.)  & Virtix-7 & 4096 & 60 & 27000/26000/31/12 & 142 & 173 & 9169  \\ \hline
\textbf{Ours (Perf-Opt.)}  & Virtix-7 & 4096 &60&	\textbf{44420/29415/384/576}&\textbf{314}&\textbf{2.5}&\textbf{184.9}\\ \hline
\textbf{Ours (Balanced)}   & Virtix-7 & 4096 &60&	\textbf{11222/7358/384/144}&\textbf{268}&\textbf{11.5}&\textbf{214.4}\\ \hline
\textbf{Ours (Area-Opt.)}  & Virtix-7 & 4096 &60&	\textbf{2133/924/384/18}&\textbf{256}&\textbf{96.1}&\textbf{293.8}\\ \hline \hline
~\cite{8675244} & Zynq UltraScale+ & 4096 & 180 & 63522/25622/400/200 & 225 & 73 & 6507.5  \\ \hline
~\cite{10171572} & Virtex-7 & 4096 & 180 & 131088/150966/576/1344 & 204 & 4.68 & 1320.0  \\ \hline
~\cite{8866755} & Virtex-7 & 4096 & 180 & 140000/-/258/1198 & 200 & 6.9 & 966  \\ \hline
~\cite{9755024} & Zynq UltraScale+ & 4096 & 128 & 46591/35551/392/400 & 250 & 24.58 & 2019.0  \\ \hline
\textbf{Ours (Perf-Opt.)} & Virtix-7 & 4096 &180 & \textbf{132608/84807/1152/1984}	& \textbf{311}& \textbf{2.5}& \textbf{550.3} \\ \hline
\textbf{Ours (Balanced)} & Virtix-7 & 4096 &180 & \textbf{35745/21247/1152/496}	& \textbf{255}& \textbf{12.1}& \textbf{692.1} \\ \hline
\textbf{Ours (Area-Opt.)} & Virtix-7 & 4096 &180 & \textbf{6221/2659/1152/62}	& \textbf{231}& \textbf{106.6}& \textbf{947.0} \\ \hline \hline
\end{tabular}
\end{table*}

\subsection{Power Analysis}
We conducted another experiment to analyze the power usage of our accelerator. Figure~\ref{fig:Power Consumption} illustrates the total power consumption with different numbers of CBUs. The observed trend indicates a decrease in overall power consumption from 4.572W to 2.465W as the number of CBUs decreases from 32 to 1, as fewer CBUs result in lower power consumption in the calculation process.

\begin{figure}
    \centering
    \includegraphics[width=1\columnwidth]{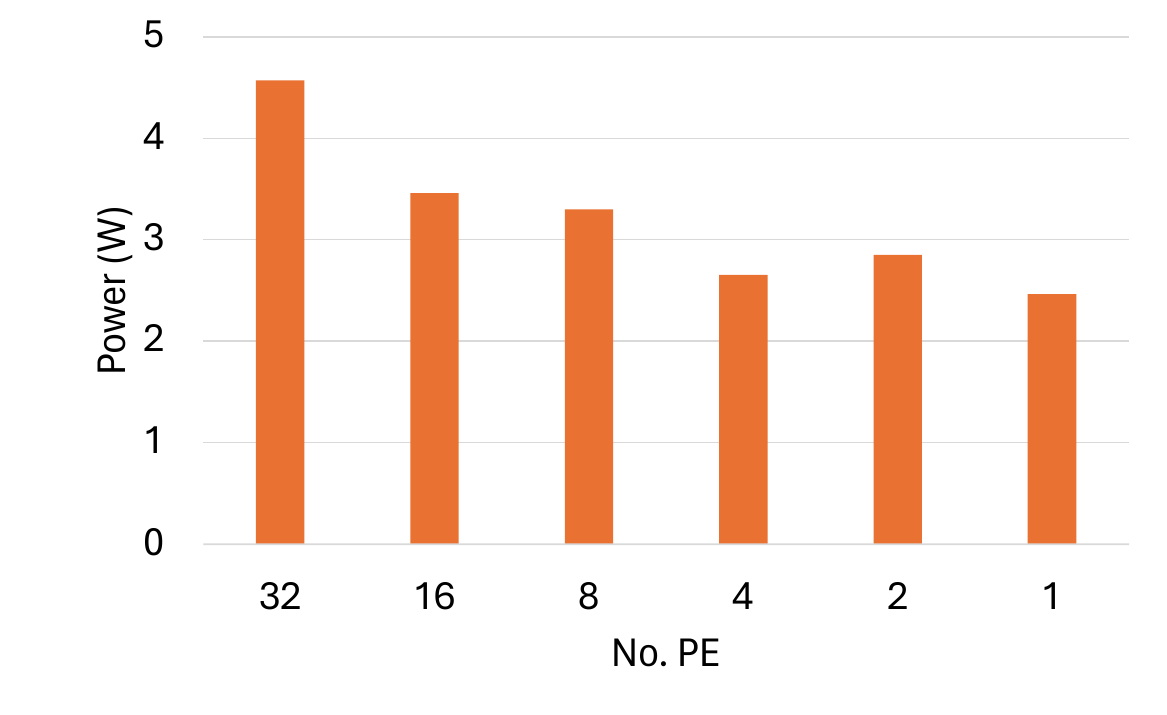}
    \caption{Overall Power Consumption with Different CBU Number $N_{pe}$ ($N = 4096, \lceil \log_2q \rceil = 32$)}
    \label{fig:Power Consumption}
\end{figure}

However, the energy consumption per operation exhibits a different trend compared to overall power consumption when the number of PEs decreases, which is shown in Figure~\ref{energy}. The energy per operation is calculated by dividing total energy consumption by the total number of butterfly operation. As the number of CBUs decreases from 32 to 1, we observed a substantial increase in energy consumption per operation, rising from 0.040 $\mu$J to 0.821 $\mu$J. This can be attributed to the longer computation time with fewer CBUs, leading to larger energy consumption from the controller and the memory. These results emphasize the importance of maximizing the CBU count to achieve optimal energy efficiency.

\begin{figure}[h]
\centering    
 \includegraphics[width=1\columnwidth]{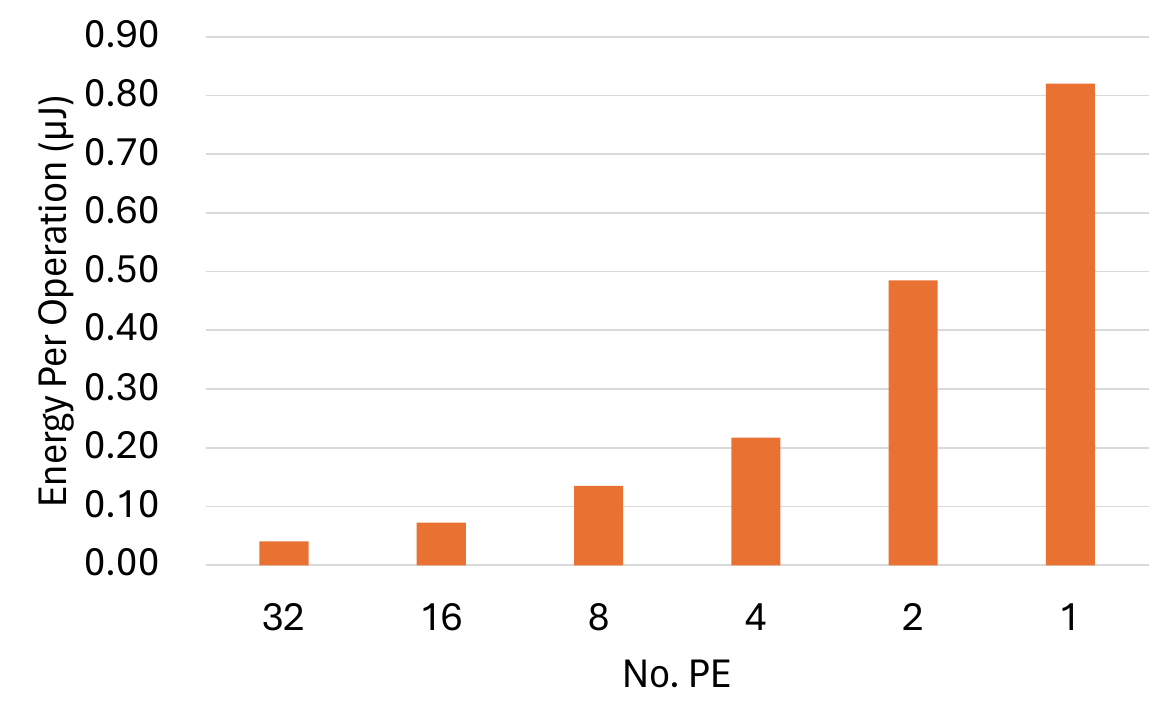}
    \caption{Energy Consumption Per Operation with Different CBU Number $N_{pe}$}
    \label{energy}
\end{figure}

\subsection{Evaluation of single-modulus polynomial multiplications}

In Table \ref{Evaluation results of our design under single-modulus and comparison with previous work}, we present a comparative analysis of key results for single-modulus settings against prior research. We focus on varying the number of CBUs, denoted as $N_{pe}$, for three scenarios: (1) $N=1024$ with $\lceil \log_2q \rceil=14$, (2) $N=4096$ with $\lceil \log_2q \rceil=32$ and (3) $N=16384$ with $\lceil \log_2q \rceil=32$. The key metrics we evaluate are latency and (LUT+FF)\_ATP.

For the $N=1024$ case with $N_{pe}=16$, our latency optimization achieves the fastest performance when the number of CBUs is equal. With 16 CBUs, it only costs $\mathbf{1.2 \mu s}$. Additionally, our (LUT+FF)\_ATP is the smallest when $N_\text{pe}$=8, at just $\mathbf{11.3}$, which improve from $1.13\times$ to $3.95\times$ compared to other baselines. It is worth noting that the clock cycles of the work~\cite{9882331} are less than ours. It is because we add more buffers before and after the CBU Array to decrease the delay in the path and optimize the frequency. Hence, in the same FPGA platform, we could obtain higher frequency and lower latency.

In the $N=4096$ scenario with $N_\text{pe}=32$, our approach achieves the latency of $\mathbf{2.7 \mu s}$, with the minimum (LUT+FF)\_ATP value of $\mathbf{121.0}$, representing an improvement ranging from $\mathbf{1.33x}$ to $\mathbf{38.61x}$ compared to the other methods.

In the study for $N=16324$, we only show the result with PE number varying from 16 to 32. Our (LUT+FF)\_ATP is \textbf{550.5}, with the reduction ratio ranging from $\mathbf{2.1x}$ to $\mathbf{117.9x}$, demonstrating our efficient resource utilization. Additionally, we were able to achieve a high frequency, although slightly below the frequency observed in the design presented in ~\cite{10.1145/3373376.3378523}, which utilized the more high-end Intel Arria10 GX 1150 board incorporating 20nm technology.

\subsection{Evaluation of multi-modulus polynomial multiplications}

Finally, we conducted experiments with parameters $\lceil\log_2q\rceil=60$ and $180$ at $N=4096$ by employing the RNS to decompose the modulus, with the results shown in  Table~\ref{Evaluation results of our design under large parameters and comparison with previous work}. In order to accommodate various operational requirements, we have established three design configurations: a performance-optimized configuration with 32 PEs, a balanced configuration with 8 PEs, and an area-optimized configuration with 1 PE. 

For $\lceil\log_2q\rceil=60$, we set $N_\text{q}=2$ for RNS decomposition. 
The performance-optimized configuration achieved the lowest (LUT+FF)\_ATP of $\mathbf{184.9}$, demonstrating a significant improvement ranging from $\mathbf{2.38x}$ to $\mathbf{49.59x}$ compared to the other methods. Additionally, this configuration also resulted in the smallest timing delay of $\mathbf{2.5 \mu s}$, representing an improvement ranging from $\mathbf{1.87x}$ to $\mathbf{69.2x}$. Even for the balanced configuration, we can achieve satisfactory time and area costs, providing a good tradeoff between performance and resource utilization.

For $\lceil\log_2q\rceil=180$, we set $N_\text{q}=6$ for decomposition and obtained the lowest (LUT+FF)\_ATP of $\mathbf{550.3}$ in the performance-optimized configuration, which is $\mathbf{1.76x}$ to $\mathbf{11.83x}$ better than the other methods. The smallest latency achieved in this configuration is $\mathbf{2.5 \mu s}$ with an improvement of $\mathbf{2.76x}$ to $\mathbf{29.2x}$.

\section{Conclusion}
\label{sec:conclude}

In this paper, we introduce HF-NTT, a hazard-free and reconfigurable NTT accelerator tailored for flexible polynomial degrees, modulus sizes, and processing element numbers. HF-NTT incorporates a novel set of data flow methodologies specifically designed for hardware implementations. These methodologies are designed to prevent different hazards during computations to reduce the number of clock cycles. Additionally, HF-NTT applies a hardware-efficient Barrett modular multiplication approach and introduces a configurable butterfly unit that can dynamically adjust its data path to process different operations. HF-NTT is implemented through Chisel and synthesized and implemented through Vivado 2022.2. The experimental evaluations demonstrate that HF-NTT achieves excellent performance and efficiency in highly parallelized designs, making it suitable for a variety of real-world FHE application scenarios.

\bibliographystyle{IEEEtranS}
\bibliography{refs}

\end{document}